\documentclass[
pra,
reprint,superscriptaddress
]{revtex4-1}

\usepackage[utf8]{inputenc}
\usepackage[T1]{fontenc}
\usepackage[]{amsmath}
\usepackage{amsfonts,braket,amsthm,mathtools,bbm,nicefrac}
\usepackage{color}
\usepackage[bookmarksnumbered]{hyperref}
\usepackage{cleveref}

\newcommand{\tg}{\tau_{\mbox{\tiny G}}}
\newcommand{\Nat}{\mathbb{N}}
\newcommand{\de}{d_{\mathrm{eff}}}
\newcommand{\dif}{\mathrm{d}}

\newcommand{\expo}[1]{\mt{e}^{#1}}
\newcommand{\half}{\frac{1}{2}}
\newcommand{\id}{\mathbbm{1}}                             

\newcommand{\ketbra}[1]{\ket{#1}\!\!\bra{#1}}

\newcommand{\mt}[1]{\mathrm{#1}}
\newcommand{\nf}[2][1]{\nicefrac{#1}{#2}}
\newcommand{\omg}{\omega_{\mbox{\tiny G}}}
\newcommand{\omu}{\omega_{\mbox{\tiny U}}}
\newcommand{\oo}[1]{\mathcal{O}\left({#1} \right)}
\newcommand{\psg}{\Psi_{\mbox{\tiny G}}}
\newcommand{\psu}{\Psi_{\mbox{\tiny U}}}
\newcommand{\rog}{\rho_{\mbox{\tiny G}}}
\newcommand{\rou}{\rho_{\mbox{\tiny U}}}

\newcommand{\tr}[1]{\operatorname{Tr}\left[ {#1} \right]} 
\newtheorem{claim}{Claim}

\crefname{figure}{Fig.}{Figures}
\crefname{appendix}{Appendix}{Appendices}
\crefname{section}{Section}{Sections}

\begin{document}
\title{Rapid spatial equilibration of a particle in a box}

\author{Artur S.L. Malabarba}
\affiliation{H.H. Wills Physics Laboratory, University of Bristol, Tyndall Avenue, Bristol, BS8 1TL, U.K.}

\author{Noah Linden}
\affiliation{School of Mathematics, University of Bristol, University Walk, Bristol BS8 1TW, U.K.}
\author{Anthony J. Short}
\affiliation{H.H. Wills Physics Laboratory, University of Bristol, Tyndall Avenue, Bristol, BS8 1TL, U.K.}

\date{\today}

\begin{abstract}
    We study the equilibration behaviour of a quantum particle in a one-dimensional box, with respect to a coarse grained position measurement (whether it lies in a certain spatial window or not).
    We show that equilibration in this context indeed takes place and does so very rapidly, in a time comparable to the time for the initial wave packet to reach the edges of the box.
    We also show that, for this situation, the equilibration behaviour is relatively insensitive to the precise choice of position measurements or initial condition.
\end{abstract}
\maketitle

Quantum systems are time-reversible and, in their time evolution, have recurrences arbitrarily close to their initial state.
This is in sharp contrast with the irreversibility found in thermodynamic processes and with the concept of evolution towards equilibrium, which is impossible in a periodic system.
Lately however, it has been found that even a closed quantum system can appear to equilibrate~\cite{Lin09, Lin10, Reimann08, Reimann10, Short11, ShortFarrelly11, Masanes11, Hutter13, Malabarba14, Goldstein14, Masanes13, Brandao12, Vinayak12, Cramer12, Kastner13, Goldstein13, TorSan13, GogEis15, Riera12, Meyer98}.

Under the weak assumption of a Hamiltonian with non-degenerate energy gaps, for a system with high-dimension, Linden \emph{et al}~\cite{Lin09, Lin10} proved that the partial state of any small subsystem (no matter how you partition the Hilbert space) stays very close to a static state for the vast majority of its time evolution.

In parallel, Reimann~\cite{Reimann08, Reimann10} showed that, for the same set of Hamiltonians, the expectation values of any `reasonable' quantum observable on a closed system also stays predominantly close to a static value, even though the state itself does not.

Building on this, Short~\cite{Short11} showed that these results apply even if one considers the specific possible outcomes of the measurement, instead of just the expectation value, by describing it as a positive operator-valued measure (POVM).
Thus, the time-evolving state of a closed system is physically nearly indistinguishable from a static state, even under careful experimental scrutiny, for the majority of time.

Since then, more effort has been devoted to studying the dynamics of the equilibration process~\cite{ShortFarrelly11, Masanes11, Hutter13, Malabarba14, Goldstein14, Masanes13, Brandao12, Vinayak12, Cramer12, Kastner13, Goldstein13, TorSan13, GogEis15, Riera12}.
Results by Short and Farrelly~\cite{ShortFarrelly11} place an upper bound on a state's distinguishability which tightens with time.
Masanes \emph{et al}~\cite{Masanes11} derive the equilibration time scale for spin systems averaged over a set of Hamiltonians.
Hutter and Wehner~\cite{Hutter13} find a clear criterion for when systems lose their initial information due to interaction with an environment.
And Malabarba \emph{et al}~\cite{Malabarba14} and Goldstein \emph{et al}~\cite{Goldstein14} show that closed systems equilibrate surprisingly fast for typical measurements (measurements based on random projectors on the Hilbert space).

Here, we aim to gain further physical insight on the precise dynamics involved in equilibration.
Showing how physical time scales---as we observe in thermodynamic systems---emerge from the underlying quantum dynamics has proven a challenge.
In fact, it has been shown that even for general nanoscale systems you can find measurements where this time scale ranges from nanoseconds to the age of the universe~\cite{Malabarba14,Goldstein14}.
Therefore, it is clear that to solve this problem one must consider scenarios where both the Hamiltonian and the measurement are physically meaningful.

For this reason, we consider the equilibration of a simple quantum system composed of a particle in a one-dimensional box, as perceived by coarse grained position measurement (whether it lies in a certain spatial window or not).
Under these conditions, we study the equilibration profile and derive the time scales by solving the time evolution.

A special characteristic of the particle in a box is that the recurrence time (for any initial condition) is at most the ground state period.
This allows us to explore the entire time evolution numerically, which would be unfeasible for systems with very large recurrence times.
These results are then compared to those cited above and it is shown that the equilibration time scale of physical measurements is fast, albeit slower than generic typical measurements.

\section{Definitions}
\label{sec:current-work}
\label{sec:aims}
\label{sec:definitions}

The specific system studied is the quantum particle inside a one-dimensional box of size $L>0$, defined by the free-particle Schrödinger equation
\begin{equation}
    \label{eq:5}
    H = \frac{p^2}{2m},
\end{equation}
where $p$ is the momentum operator, and the boundary conditions restrict the state to the domain $x \in [\frac{-L}{2},\frac{L}{2}]$ in position space,
\begin{equation}
    \label{eq:46}
    \braket{x|\Psi}|_{x=-\frac{L}{2} } = \braket{x|\Psi}|_{x=\frac{L}{2} } = 0.
\end{equation}

We consider here a class of projective measurements $M_w$, with $w \in [0,L]$.
Each measurement ${M_w} = \{A_w, B_w\}$ is a set of two projectors in position space, each corresponding to a different outcome, defined by
\begin{align}
  \label{eq:3-1}
  A_w = \int_{- \frac{w}{2} }^{\frac{w}{2}} \ketbra{x} \dif x, \qquad B_w = \id - A_w.
\end{align}
They represent whether or not the particle is inside a region of width $w$ centered on the origin.

Given a system of dimension $d$, it can appear to equilibrate with respect to a particular measurement because all the information contained in the system's degrees of freedom (up to $d$ complex numbers) is reduced to the probabilities of each measurement outcome (one real number for each possible outcome).
Under this amount of ignorance, even orthogonal states may be indistiguishable and the system may appear to be stationary.
Choosing a two-outcome measurement ensures we are taking full leverage of this and taking an approach that is as simple as possible to understanding equilibration.

Our prime object of study is the distinguishability between two states according to a measurement $M$
\begin{equation}
    \label{eq:7}
    D_M(\rho_1,\rho_2) = \frac{1}{2} \sum_{A \in M} \left| \tr{A(\rho_1 - \rho_2)} \right|,
\end{equation}
which, for $M_w$ reduces to
\begin{equation}
    \label{eq:8}
    D_{M_w}(\rho_1,\rho_2) = \left| \tr{A_w(\rho_1 - \rho_2)} \right|.
\end{equation}
This quantity is significant because it determines the probability of successfully determining whether the system was originally in state $\rho_1$ or $\rho_2$ after performing the measurement $M$~\cite{Short11}.
This probability is $p = \half + \half D_M(\rho_1,\rho_2)$, assuming one has full knowledge of the two states being compared.

The verification that equilibration has happened (at some time $t$) is done by comparing the instantaneous state of a system $\rho(t)$ with its time-average state $\omega = \braket{\rho(t)}_{t\in (0,\infty)}$, also called the equilibrium state, where
\begin{equation}
    \label{eq:42}
    \braket{f(t)}_{t\in (0,\infty)} =
    \lim_{T\rightarrow \infty} \frac{1}{T} \int_{0}^{T} f(t) \dif t.
\end{equation}
Equilibration is achieved once $D_M(\omega,\rho(t))$ becomes small ($\ll 1$) and stays small for the majority of the time evolution.

Another quantity which plays an important role in the equilibration process is the effective dimension~\cite{Malabarba14,Short11,ShortFarrelly11}.
Roughly, it represents the number of energy eigenstates ($\ket{n}$) which the system occupies with significant probability and, for a pure state $\ket{\Psi}$, it can be written as
\begin{equation}
    \label{eq:19}
    \de = \left[ \sum_{n = 1}^{\infty} \left| \braket{\Psi|n} \right|^4 \right]^{-1}.
\end{equation}
Throughout this paper, whenever we mention the high-dimensional limit we are referring to $\de \gg 1$.

Our results, presented in the next \namecref{sec:results}, consist of mathematical statements regarding important characteristics of the equilibration process of the particle in a box system with respect to the simple dichotomic measurement described in \cref{eq:3-1}.
We are effectively asking the question: \emph{``If all I know is which interval the particle is in, does it look like it equilibrates? How fast does that happen?''}.

\section{Results}
\label{sec:results}\label{sec-1-3}\label{sec-1-3-1}\label{sec:gauss-distr}

First, we consider the initial condition with an approximately Gaussian distribution,
\begin{align}
  \label{eq:2}
  \braket{x|\psg(0)} &= N \left( \expo{-\left( \frac{x}{2\sigma} \right)^2} -
                       \expo{-\left( \frac{L}{4\sigma} \right)^2 } \right) \notag\\
                     &\approx \left( \frac{1}{2 \pi\sigma^2} \right)^{\frac{1}{4} }
                       \expo{-\left( \frac{x}{2\sigma} \right)^2 }\\
  \rog &= \ketbra{\psg},\notag
\end{align}
where the arbitrary constant $\sigma$ defines the initial width of the wave packet in position space, $N$ is the normalization constant, and the approximate relation is due to neglecting the second exponential on the first line.
This is a smooth function concentrated mostly inside the support of a single outcome of $M$ ($A_w$), as long as $\sigma < w$.
For this system, the eigenstates, the energy levels, and the energy amplitudes are
\begin{align}
  \label{eq:20}
  \braket{x|n} &= \sqrt{\frac{2}{L}} \sin \left( \! n \pi \! \left( \frac{x}{L} + \half \right) \!\!\right),
                 \quad x \in \left[ \frac{-L}{2},\frac{L}{2} \right]\\
  E_n &= \frac{\hbar^2 n^2 \pi^2}{2 m L^2} \\
  \braket{\psg|n} &\approx \left( {\frac{2\pi \sigma^2}{L^2} } \right)^{\frac{1}{4} }
                    \expo{-\left( \frac{n \pi\sigma}{L} \right)^2}
                    2 \sin{\left( \frac{n\pi}{2} \right) }.
\end{align}

Before describing the distinguishability we state our first result about this system.
The effective dimension, defined in \cref{eq:19}, can be written as
\begin{align}
  \label{eq:1}
  \de^{-1}
  &\approx \sum_{n = 1}^{\infty} 2 \left( \frac{4 \sqrt{\pi}\sigma}{L} \right)^2
    \expo{-\left( \frac{2n \pi \sigma}{L} \right)^2} \sin^4{\left( \frac{n\pi}{2} \right) } \notag\\
  &= \left( \frac{4\sqrt{\pi}\sigma}{L} \right)^2
    \sum_{n=-\infty}^{\infty} \expo{-\left( \frac{2n \pi \sigma}{L} \right)^2} \sin^4{\left( \frac{n\pi}{2} \right) } \notag\\
  & = \left( \frac{4\sqrt{\pi}\sigma}{L} \right)^2
    \sum_{j=-\infty}^{\infty} \expo{-\left( \frac{4 \pi \sigma}{L} \right)^2\left(j + \nf{2} \right)^2} \notag\\
  & = \left( \frac{4\sqrt{\pi}\sigma}{L} \right)^2 \left( \int_{-\infty}^{\infty} \!\expo{-\left( \frac{4 \pi \sigma}{L} \right)^2\left(j + \nf{2} \right)^2} \dif j + \oo{\frac{\sigma}{L}} \right) \notag\\
  & = \left( \frac{4\sqrt{\pi}\sigma}{L} \right) \left( 1 + \oo{\frac{\sigma^2}{L^2} } \right)
\end{align}
Where in the third line we have used the fact that only the odd-$n$ terms contribute to the sum.
Furthermore, since the error involved in the approximation of \cref{eq:2} is exponentially small in $L/\sigma$, for $\sigma/L$ small enough we may incorporate this into $\oo{\frac{\sigma^2}{L^2}}$, which leads to
\begin{equation}
    \label{eq:13}
    \de = \frac{L}{4\sqrt{\pi}\sigma} + \oo{\frac{\sigma}{L} }.
\end{equation}
This means that $\de \gg 1$ corresponds to $\sigma \ll L$.

\subsection{The Distinguishability}
\label{sec:distance}

The equilibrium state~\cite{ShortFarrelly11,Short11,Lin09} is calculated as the average over all time of the initial condition.
This is equal to the initial state decohered in the energy basis, which, for $\rho(0) = \rog$ reduces to
\begin{align}
  \label{eq:10}
  \omg \approx \sqrt{\frac{\pi}{2} }{\frac{8\sigma}{L} }\sum_{n = 1}^{\infty}
  \ketbra{n} \expo{-2{\left( \frac{n \pi \sigma}{L} \right)}^2}
  \sin^2{\left( \frac{n\pi}{2}\right) },
\end{align}
where, again, the error is exponentially small in $\frac{L}{\sigma}$, as per \cref{eq:2}.
Naively, one might expect this to have a uniform distribution in position space, but \cref{fig:4} shows that there remains a narrow `spike' at the origin even in the high dimensional limit.
This arises because $\omg$ only contains energy levels of odd $n$, all of which have maximal probability at the origin.
\begin{figure}
    \includegraphics[width=\linewidth]{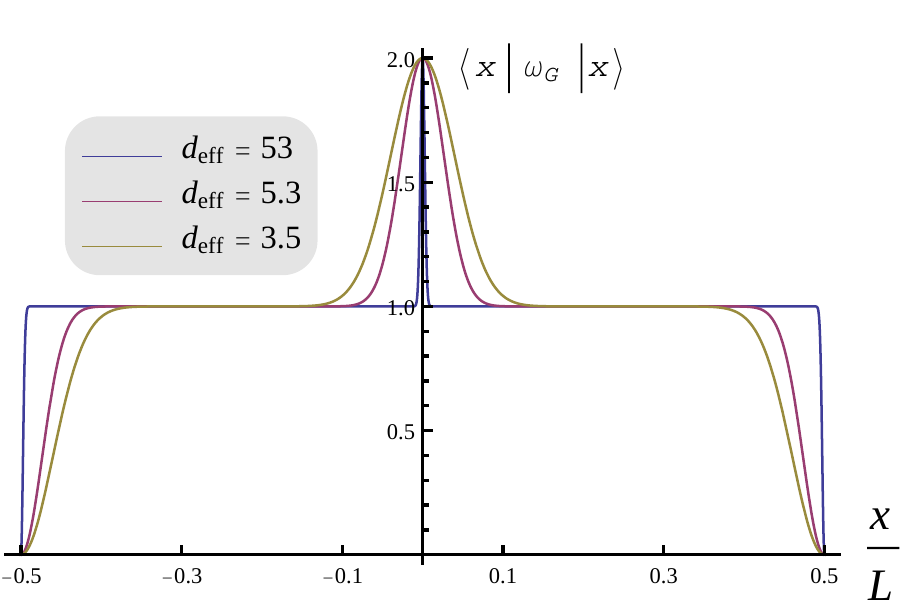}
    \caption{Probability distribution of the equilibrium state, $\braket{x|\omg|x}$ from \cref{eq:10}, of the Gaussian initial condition \cref{eq:2} for several values of $\de$.}
    \label{fig:4}
\end{figure}

Finally, armed with a time-evolving state $\rog(t)$ and its time-averaged state $\omg$ we can investigate the equilibration process.
\begin{figure}
    \centering
    \includegraphics[width=\linewidth]{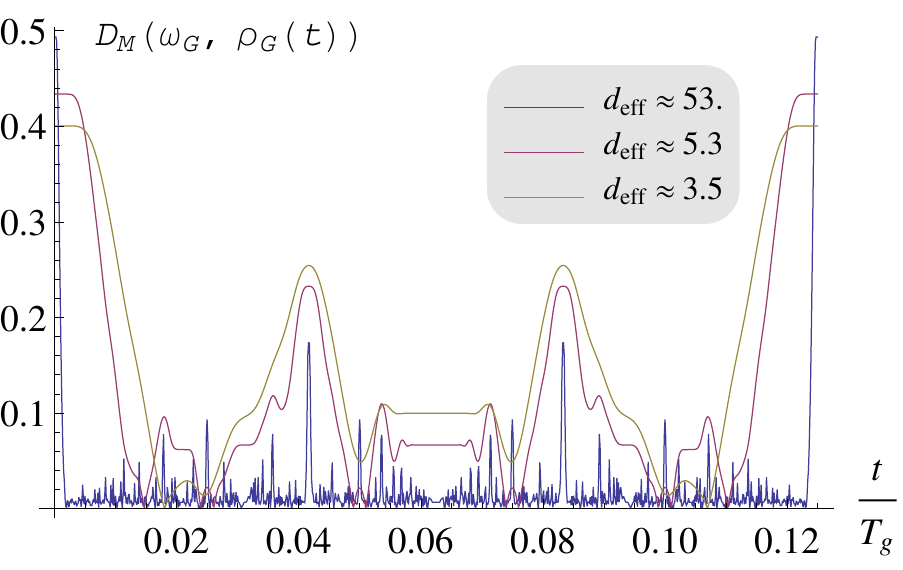}
    \caption{Distinguishability, as defined in \cref{eq:8}, between the time evolution of $\rog$ (\cref{eq:2}) and its equilibrium state $\omg$ as a function of time, for three different $\de$ and $w = \frac{L}{2}$.
    $T_g$ is the ground state period, see \cref{eq:4}.}
    \label{fig:3}
\end{figure}
For the system and measurement in question, the time evolution of $D_{M_w}(\omega,\rho)$, as defined in \cref{eq:8}, is displayed in \cref{fig:3} for $w = L/2$.
The effect of equilibration can be seen clearly for high dimension, as the system (i) starts sharply out of equilibrium, (ii) quickly drops to a small fraction of the initial distance, and (iii) oscillates within that distance for the predominant majority of the evolution.
It is also remarkable that equilibration happens with such ease, where only extremely simple systems ($\de < 50$) are well distinguishable from their equilibrium states.

One common pattern in the literature~\cite{ShortFarrelly11, Malabarba14, Lin09, Short11} is to analyse how the infinite-time average of the distinguishability scales with the effective dimension of the initial state.
\begin{figure}
    \includegraphics[width=\linewidth]{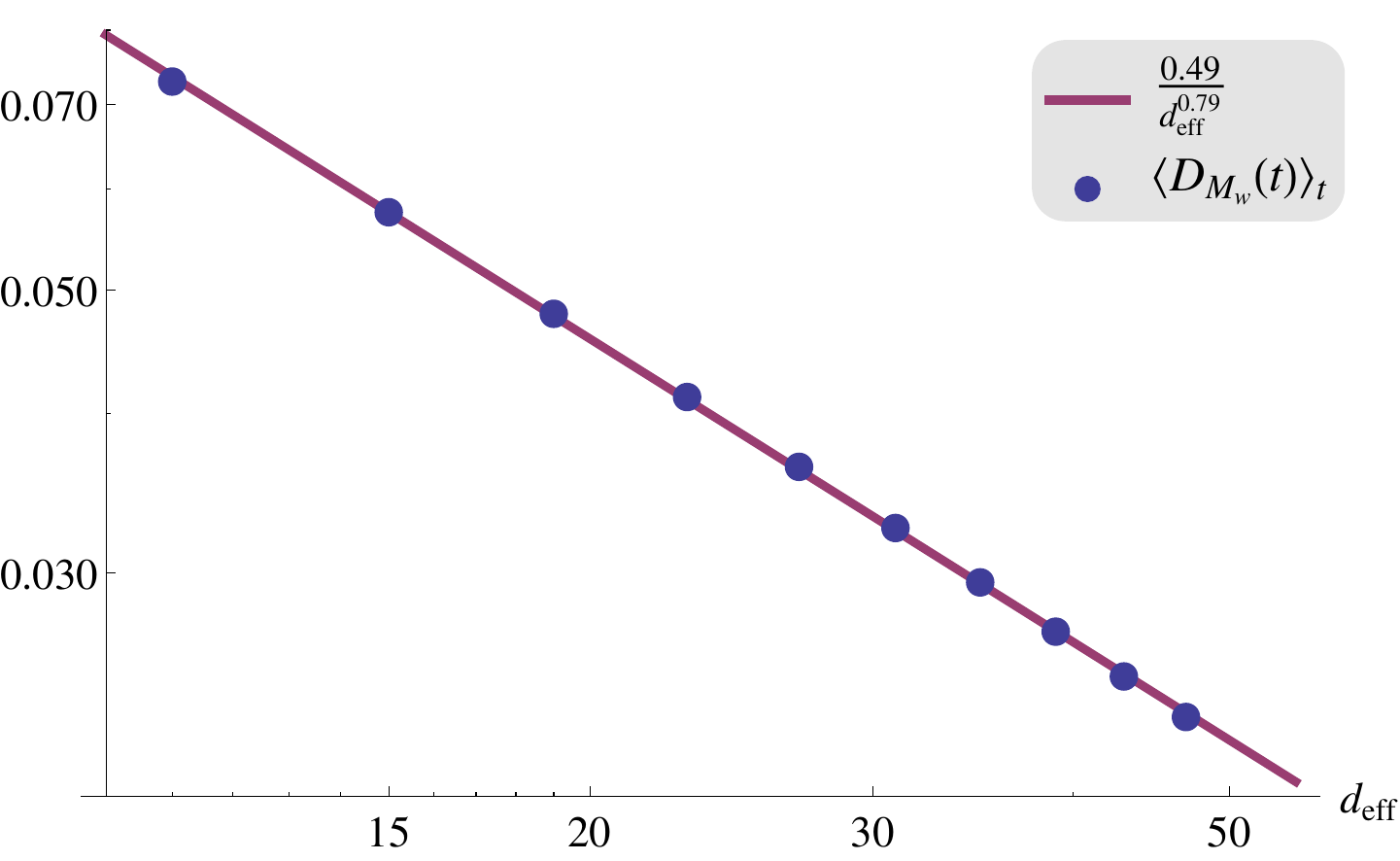}
    \caption{Log-scale plot of the time average of the distinguishability for the Gaussian initial condition, $D_{M_\frac{L}{2}}(\rog(t),\omg)$, as per \cref{eq:8,eq:10,eq:2}, for several values of $\de$.
    The line shows the best linear fit over the logarithmic data.}
    \label{fig:6}
\end{figure}
For our conditions, \cref{fig:6} illustrates that it scales inverse-polynomially, at roughly $1/2\de^{0.8}$, slightly faster than the known upper bounds for general systems~\cite{ShortFarrelly11}, which scale with $\de^{-1/2}$.
This is likely because these bounds always involve a step employing Jensen's inequality, $\braket{\sqrt{D^2}}<\sqrt{\braket{D^2}}$, which probably underestimates the scaling.

\subsection{The Time Scale}
\label{sec:time-scale}

From the evolution of $D_{M_w}(\rog(t),\omg)$ it is easy to see exactly when the high-dimensional system equilibrates, given that after the initial decay it displays only a few short-lived fluctuations away from equilibrium.
The red dots of \cref{fig:2-1} (the same data as the blue line of \cref{fig:3} for a shorter time interval) display the equilibration profile quite clearly.
Thus, for this system, the equilibration time scale can be defined as the characteristic width of the first peak.
\begin{figure}
    \includegraphics[width=\linewidth]{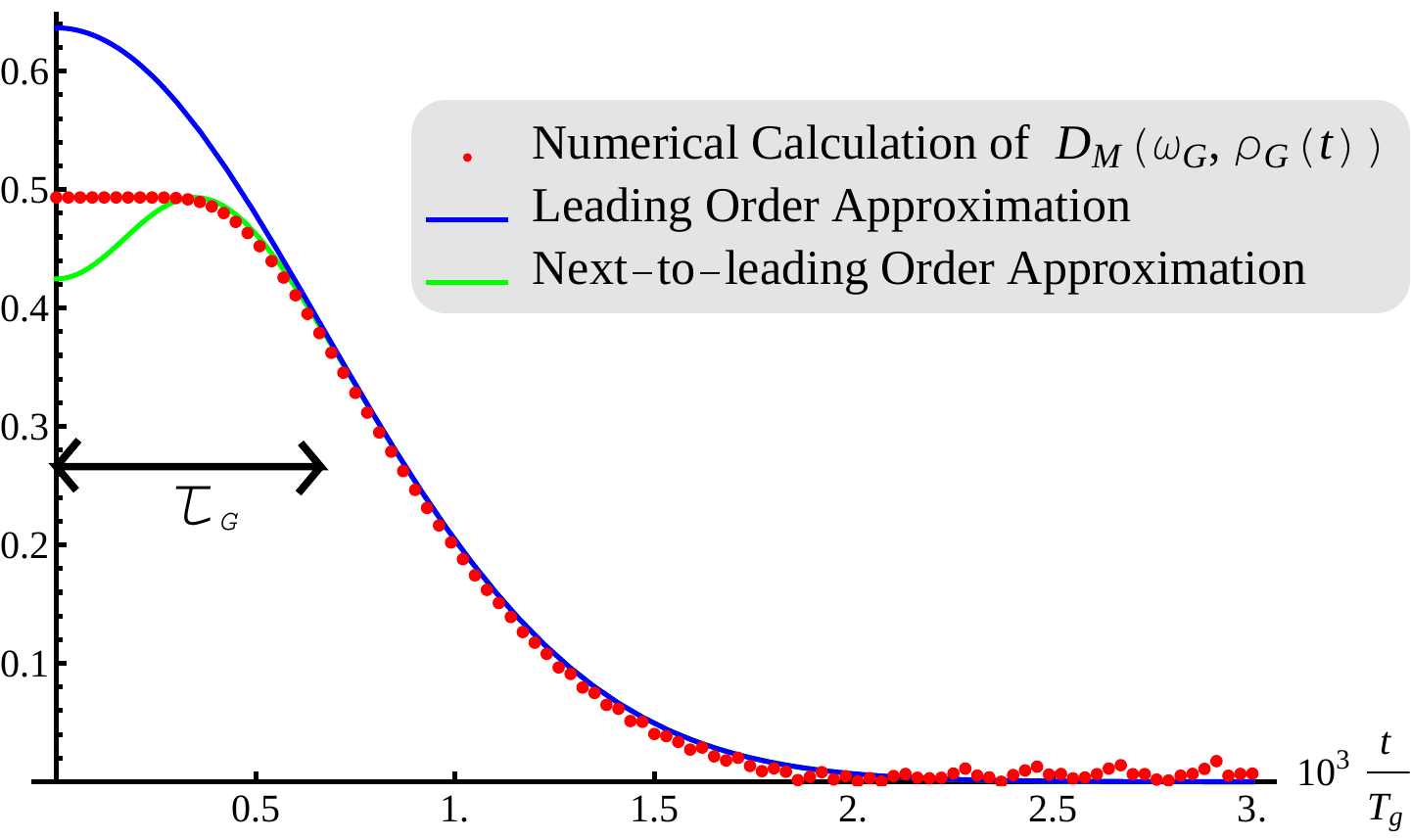}
    \caption{Numerical computation $D_{M_\frac{L}{2}}(\rho(t),\omega)$ compared to the leading order (\cref{eq:12}) and next-to-leading order approximations, for $\de \approx 53$.
      The expression for the next-to-leading order, as well as the full expansion, is found in \cref{eq:42-app} of \cref{sec:equil-time-scale}.
      The red dotted line is a zoom of the plot in \cref{fig:3}.
      The constant region up to $t/T_g \approx 3\times 10^{-4}$ happens because the wave function has not yet spread beyond the edges of the measurement window.}
    \label{fig:2-1}
\end{figure}

This leads us to our most important result of this section.
As detailed in \cref{sec:equil-time-scale}, in the limit $\de \gg 1$ and $t \ll \tg \de$, \cref{eq:8} reduces to
\begin{align}
  \label{eq:12}
  \!\!\!D_{M_{\frac{L}{2} }}(\omg,\rog(t)) = \frac{2}{\pi} \left| \expo{-\frac{t^2}{2\tg^2} }
  + \oo{\!\!\!{\left( \expo{-\frac{t^2}{2\tg^2} } \right)}^9,\frac{\sigma}{L} } \right|,
\end{align}
where
\begin{equation}
    \label{eq:3}
    \tg = \frac{mL \sigma}{\hbar \pi},
\end{equation}
is the characteristic time scale for equilibration in this system.
The significance of this expression is clear from \cref{fig:2-1}, which compares the same numerical results presented in \cref{fig:3} (for the high dimension case) with the expression above, in the short time limit where equilibration occurs.
For $t\approx 0$ you need higher order terms to get a precise estimate.
However, for $t\approx \tg$, the lowest order approximation becomes the dominant contribution, and it is all one needs to describe the distinguishability in the vicinity of the equilibration time scale.

One way of writing \cref{eq:3} adimensionally is to compare it to the period of the groundstate $T_g$
\begin{equation}
    \label{eq:4}
    \frac{\tg}{T_g} = \frac{1}{16 \sqrt{\pi}\de},
\end{equation}
where $T_g = \frac{2\pi\hbar}{E_1} = \frac{4mL^2}{\hbar \pi}$.

It is also interesting to compare this with recent results~\cite{Malabarba14}, which show that the equilibration time scales for typical measurements (random pojectors on the Hilbert space) scale with $\hbar/SD_E$, where $SD_E$ is the standard deviation in energy of the state.
In \cref{sec:energy-variance}, we show that this corresponds to
\begin{align}
  \label{eq:41}
  \frac{\tau_{\mathrm{typical}}}{T_g} &= \frac{1}{16 \de^2}
\end{align}
for our system.
This proves that, while our system presents fast equilibration for projective position measurements, it is not nearly as fast as the time scales predicted for random projective measurements.
This corroborates the notion that measurements of physical significance (like the ones studied here) are much slower than the average.

\section{Generality}
\label{sec:discussion}

\subsection{On the Measurements}
\label{sec:measurements}

\begin{figure}
    \includegraphics[width=\linewidth]{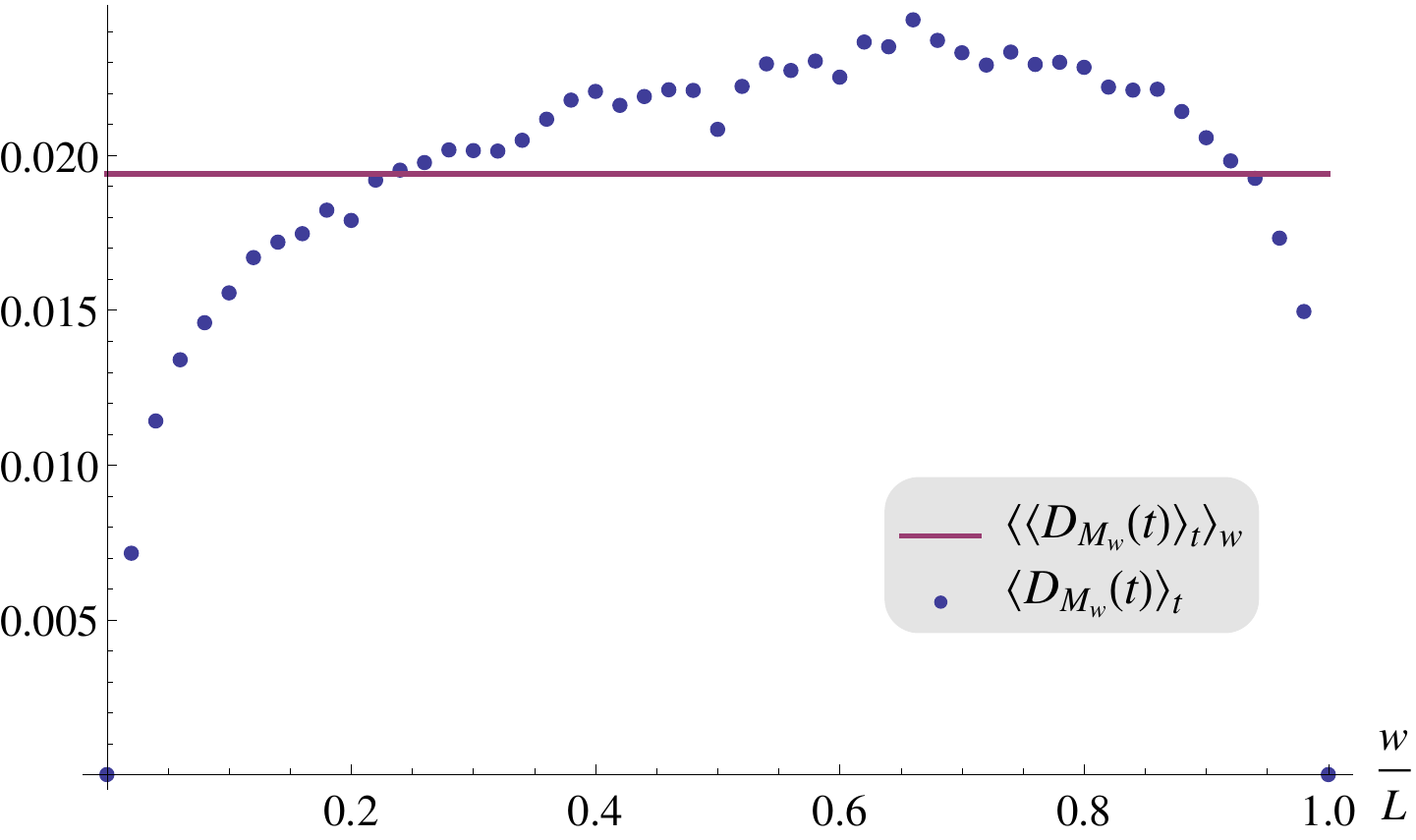}
    \caption{Numerical computation of the distinguishability, $D_M(\omg,\rog(t))$, averaged over time, for a range of different measurement widths.
      The line is the average value of this distribution.
      Initial condition is a Gaussian wave-packet \cref{eq:2} with $\de \approx 53$.}
    \label{fig:5}
\end{figure}
All the plots in the previous Section take $w$ to be $L/2$.
However, numerical tests were performed with various different widths and neither the overall behaviour of $D_{M_w}(\omega,\rho)$ nor the equilibration time scale varied significantly.
As seen in \cref{fig:5}, the amplitude of $D_{M_w}$ only shows significant change when $w$ is very close to $0$ or $L$---in which case it falls to $0$ as the measurement conveys no information.
This is evidence that these results are not a consequence of some particular symmetry emerging from the chosen width, but are the behaviour one can expect from projective measurements in position space.
In fact, most measurement of this class are slightly better than the $w=L/2$ measurement, but this choice of $w$ was helpful for our analytical results.

\subsection{On The Initial Condition}
\label{sec:other-conditions}

In order to verify that our results are not a consequence of some special property of the initial condition, the same Hamiltonian and measurement were also investigated under the initial condition of a uniform distribution in energy space
\begin{equation}
    \label{eq:11}
    \ket{\psu} = \frac{1}{\sqrt{\de}}\sum_{i = 1}^{\de} \ket{i}.
\end{equation}
Satisfactorily, this system (described in more detail in \cref{sec-1-1}) presented a very similar equilibration profile.
Its time scale was analytically estimated to be similar to the Gaussian scenario in its proportionality to $\de^{-1}$,
\begin{equation}
    \label{eq:21}
    \frac{\tau_{\mbox{\tiny U}}}{T_g} \approx \frac{1}{4 \de},
\end{equation}
and the infinite-time average of its distinguishability scales with $0.77\de^{-0.77}$, which is also very similar to the Gaussian's scaling although the multiplicative constant is larger.

\section{Discussion}
\label{sec:discussion-1}

Recent developments in the area have shown that generic time scale bounds are too weak to be physically meaningful, while typical (but non-physical) measurements or typical Hamiltonians equilibrate extremely rapidly.
As such, in the search for physical equilibration time scales, we must investigate the specifics of particular systems and measurements.
This work takes as step in that direction, thoroughly analysing a physical system and showing that fast time scales emerge, albeit much slower than those predicted for a typical observable.

Here, we have exposed the equilibration dynamics of a wave-packet constrained inside a box when inspected by coarse-grained position measurements (onto a spatial window).
We have derived an equation for the distinguishability in the vicinity of the equilibration time scale, which also results in an equation for the time scale itself.
This, in turn, is used to show that it scales with $1/\de$, in constrast with the much faster time of $1/\de^2$ which is known for typical measurements in this system.

Note that an estimate of the average time for the particle to reach the edges of the box is given by
\begin{equation}
    \tau_{box} = \frac{(L/2)} {(\Delta p/m)}
    = \frac{m L \sigma} {\hbar} = \pi \tg.
\end{equation}
Hence, perhaps surprisingly, the time scale for equilibration is comparable to the time scale to reach the walls for the first time.
In particular, this means that equilibration is fast in a natural sense, and the particle does not need to `bounce back and forth' many times in the box in order to equilibrate well.

It is remarkable that, despite being a simple and well-studied textbook system, there are still questions to be asked about the particle in a box and still lessons to be learned.
In particular, our results show that this simple physical system can display equilibration behaviour over a physically realistic time scale, which reinforces recent developments concerning equilibration of abstract systems.

\bibliography{../../../Org/papers.org}
\appendix

\section{Equilibration Time Scale}
\label{sec:equil-time-scale}
In this \namecref{sec:equil-time-scale} we use the following abbreviations, $\phi = 2\pi \sigma/L$, $\nu_{nj} = (E_n - E_j)/\hbar$, $\nu = E_1/\hbar$, and $\rho_{n,j} = \braket{n|\rog|j}$.
\begin{claim}\label{cl:3}
    According to the set of outcomes $M = \{A_w,B_w\}$, with $w = \frac{L}{2}$, the time-dependent distance between the state in \cref{eq:2} and its equilibrium state (\cref{eq:10}) is
    \begin{equation}
        \label{eq:41-app}
        D_M(\omg,\rog(t)) \approx \frac{2}{\pi}\frac{\phi}{\sqrt{2\pi}}
        \left| \sum_{k,j=1}^{\infty}
          \cos{( 4kl\nu t )} \expo{-(k^2+l^2)\frac{\phi^2}{2}} B_{kl}
        \right|,
    \end{equation}
    where $B_{kl}= \left[ \cos{l \pi } - \cos{k \pi } \right] \left[\frac{\sin(k\pi/2)}{k} - \frac{\sin(l\pi/2)}{l}\right]$.
    Where the approximate relation is exponentially precise in $1/\phi$, as per \cref{eq:20}.
\end{claim}
\begin{proof}\label{pr:1}
    First, as defined in \cref{eq:8}
    \begin{align}
      \label{eq:34-app}
      D_M(\omg,\rog(t))
      &= \left| \tr{A_\frac{L}{2}(\rog(t)-\omg)} \right| \\
      &= \left| \sum_{n\neq j = 1}^{\infty} \rho_{n,j} \expo{i \nu_{nj} t} \braket{n|A_\frac{L}{2}|j}\right|\notag\\
      &= 2 \left|\sum_{n = 2}^{\infty} \sum_{j=1}^{n-1} \rho_{n,j} \cos( \nu_{nj} t) \braket{n|A_\frac{L}{2}|j} \right|, \notag
    \end{align}
    where
    \begin{align}
      \label{eq:25}
      \rho_{n,j} &\approx \frac{4 \phi}{\sqrt{2\pi}}
                   \expo{-\phi^2 \frac{j^2+n^2}{4} }
                   \sin \frac{n\pi}{2}\sin \frac{j\pi}{2} \\
                 &= \frac{2 \phi}{\sqrt{2\pi}}
                   \expo{-\phi^2 \frac{j^2+n^2}{4} }
                   \left[ \cos \frac{\pi(n-j)}{2} - \cos \frac{\pi(n+j)}{2} \right].\notag
    \end{align}
    Furthermore, it is straightforward to see that
    \begin{align}\label{eq:35}
      \braket{n|A_\frac{L}{2}|j}
      &= \frac{2}{L} \int_{-\frac{L}{4}}^{\frac{L}{4}}
        \sin \left( \! n \pi \! \left( \frac{x}{L} + \half \right) \!\!\right)
        \sin \left( \! j \pi \! \left( \frac{x}{L} + \half \right) \!\!\right) \dif x \notag\\
      &= \frac{2}{L} \int_{\frac{L}{4}}^{\frac{3L}{4}}
        \sin \left( n \pi x/L \right)
        \sin \left( j \pi x/L \right) \dif x \notag\\
      &= \frac{2}{\pi} \int_{\frac{\pi}{4}}^{\frac{3\pi}{4}}
        \sin \left( n \alpha \right)
        \sin \left( j \alpha \right) \dif \alpha \notag\\
      &= \frac{1}{\pi} \int_{\frac{\pi}{4}}^{\frac{3\pi}{4}}
        \cos \left( (n-j) \alpha \right)
        - \cos \left( (n+j) \alpha \right) \dif \alpha \notag\\
      &= \frac{1}{\pi} \left[
        \frac{\sin \left( n-j \right)}{n-j} -
        \frac{\sin \left( n+j \right)}{n+j}
        \right]_{\frac{\pi}{4} }^{\frac{3\pi}{4} } \notag\\
      &= -\frac{(-1)^{n-j} + 1}{\pi ({n-j})} {\sin \left( \pi \frac{n-j}{4} \right)} \notag\\
      &\quad\quad +\frac{(-1)^{n+j} + 1}{\pi ({n+j})} {\sin \left( \pi \frac{n+j}{4} \right)}.
    \end{align}
    Finally, we state the identity
    \begin{align}
      \label{eq:38-app}
      \sum_{n = 2}^{\infty} \sum_{j=1}^{n-1} f(n,j) = \sum_{k=2}^{\infty}& \sum_{l=1}^{k-1}f(k+l,k-l) + \notag\\
                                                                         &\sum_{k=1}^{\infty} \sum_{l=0}^{k-1}f(k+l+1,k-l),
    \end{align}
    and note that
    \begin{equation}
        \label{eq:27}
        \sin \frac{(k+l+1)\pi}{2}\sin \frac{(k-l)\pi}{2} = 0, \qquad \forall k,l \in \mathbb{Z},
    \end{equation}
    and
    \begin{align}
      \label{eq:43}
      \braket{k+l|A_\frac{L}{2}|k-l} &= \frac{1}{\pi} \left[ \frac{\sin(\pi k/2)}{k} - \frac{\sin(\pi l/2)}{l} \right].
    \end{align}
    Combining all of these identities with \cref{eq:34-app}, yields
    \begin{equation}
        \label{eq:22}
        D_M(\omg,\rog(t)) \approx \frac{4}{\pi}\frac{\phi}{\sqrt{2\pi}}
        \left| \sum_{k=2}^{\infty}
          \sum_{l=1}^{k-1} \cos{( 4kl\nu t )}
          \expo{-(k^2+l^2)\frac{\phi^2}{2}} B_{kl}
        \right|,
    \end{equation}
    which can be easily reduced to \cref{eq:41-app} by noting that the summand is symmetric with respect to $k$ and $l$, and that $B_{kk}=0$.
\end{proof}

Next, note that the double sum in \cref{eq:41-app} can be divided into two expressions, respectively proportional to $\sin(l\pi/2)$ and to $\sin(k\pi/2)$.
Furthermore, $\sin(l\pi/2)$ is non-zero only when $l$ is odd, which means $(\cos{l \pi } - \cos{k \pi }) = -(1 + \cos{k \pi })$ which, in turn, means that the only non-zero terms of this expression those with odd $l$ and even $k$.
Then, the part of the double sum in \cref{eq:22} which is proportional to $\sin(l\pi/2)$ looks like
\begin{equation}
    \label{eq:48}
    - 2\sum_{\begin{subarray}{c}k=2\\\mathrm{even}\end{subarray}}^{\infty} \sum_{\begin{subarray}{c}l=1\\\mathrm{odd}\end{subarray}}^{\infty} 
    \cos{( 4kl\nu t )} \expo{-(k^2+l^2)\frac{\phi^2}{2}} \left[ \frac{\sin(l\pi/2)}{l} \right].
\end{equation}
Applying the same logic to the expression proportional to $\sin(k\pi/2)$ yields
\begin{equation}
    \label{eq:48-1}
    - 2\sum_{\begin{subarray}{c}k=1\\\mathrm{odd}\end{subarray}}^{\infty}\sum_{\begin{subarray}{c}l=2\\\mathrm{even}\end{subarray}}^{\infty} 
    \cos{( 4kl\nu t )} \expo{-(k^2+l^2)\frac{\phi^2}{2}} \left[ \frac{\sin(k\pi/2)}{k} \right],
\end{equation}
which is identical to \cref{eq:48} up to a swap of dummy indices.
Thus, by writing the odd index in each expression as $2p+1$ and the even index as $2o$, we have
\begin{align}
  \label{eq:41-app-2}
  &D_M(\omg,\rog)\\
  &\quad\approx \frac{8}{\pi} \frac{\phi}{\sqrt{2\pi}}
    \Bigg| \sum_{{\begin{subarray}{c}o=1\\p=0\end{subarray}}}^{\infty} \cos{(8o(2p+1)\nu t )} \frac{(-1)^p}{2p+1}\notag\\[-0.2cm]
  &\qquad\qquad\qquad \qquad
    \expo{-(4o^2+(2p+1)^2)\frac{\phi^2}{2}} \Bigg|.\notag
\end{align}
We can further abbreviate this by defining $R_p = \frac{{(-1)}^p}{2p+1} \expo{-{(2p+1)}^2 \frac{\phi^2}{2} }$,
\begin{align*}
  D_M(\omg,\rog) &\approx \frac{8\phi}{\pi\sqrt{2\pi}} \left| \sum_{{\begin{subarray}{c}o=1\\p=0\end{subarray}}}^{\infty}
  \cos{\left(8(2p+1)o\nu t\right)}R_p\expo{-2o^2\phi^2} \right|\\
                 &= \frac{2}{\pi} \frac{4 \phi}{\sqrt{2\pi}}\! 
                   \left| \sum_{p=0}^{\infty} \! R_p\!\!
                   \sum_{o=1}^{\infty} \!\! \expo{-2o^2\phi^2} \!\! \cos{\left(8(2p+1)o\nu t\right)}
                   \right|.
\end{align*}
This, combined with the small time assumption, allows us to approximate the sum over $o$ by an integral
\begin{align*}
  \frac{4 \phi}{\sqrt{2\pi}} \sum_{o=1}^{\infty} \cos{\left(8(2p+1)o\nu t\right)}\expo{-2o^2\phi^2}
  &= \expo{-\frac{t^2{(2p+1)}^2}{2\tg^2} } + \oo{\phi},
\end{align*}
where $\tg$ is defined in \cref{eq:3}.
This directly leads to
\begin{align}
  \label{eq:42-app}
  D_M(\omg,\rog(t)) &= \frac{2}{\pi} \left| \sum_{p=0}^{\infty} R_p \expo{-\frac{t^2{(2p+1)}^2}{2\tg^2} } + \oo{\frac{\sigma}{L}} \right|,
\end{align}
which then produces \cref{eq:12}.
In particular, note that $|R_1| = |R_0|^9/3$, which is why the $p=0$ term is sufficient for a very good estimate of the equilibration time.
On the other hand, at $t=0$ this is very close to an alternating harmonic series, which is why it converges slowly and needs many terms to describe the constant region from \cref{fig:2-1}.

\section{Energy Standard Deviation}
\label{sec:energy-variance}

In this Appendix, we calculate the energy standard deviation (named here $v_E$) for the Gaussian initial state.
Taking $\phi = 2\frac{\pi\sigma}{L} $, we have, for any $p \in \Nat$,
\begin{align}
  \label{eq:30}
  \tr{H^p\rog} &= \sum_{n = 1}^{\infty} \braket{n|\rog|n} E_n^p \notag\\
               &\approx E_1^p \frac{4\sigma \sqrt{2\pi}}{L} \sum_{n = 1}^{\infty} n^{2p} \expo{-n^2 \frac{\phi^2}{2} } \sin^2 \left( \frac{n \pi}{2} \right) \notag\\
               &= E_1^p \frac{2 \phi}{\sqrt{2\pi}} \sum_{n = -\infty}^{\infty} n^{2p} \expo{-n^2 \frac{\phi^2}{2} } \sin^2 \left( \frac{n \pi}{2} \right) \notag\\
               &= \left( \frac{2 E_1}{\phi^{2}} \right)^p S_p(\sqrt{2}\phi),
\end{align}
where
\begin{equation}
    \label{eq:34}
    S_p(\gamma) = \frac{1}{\sqrt{\pi}} \sum_{j = -\infty}^{\infty} \gamma^{2p+1}{(j+\nf{2})}^{2p} \expo{-{(j+\nf{2})}^2 \gamma^2 }.
\end{equation}

Given that
\begin{equation}
    \label{eq:31}
    S_p(\gamma) = \frac{1}{\sqrt{\pi}}\int_{- \infty}^{\infty} \!\!\dif x\,\, x^{2p} \expo{-x^2} + \oo{\gamma^2},
\end{equation}
one has that, for small $\gamma$,
\begin{align}
  \label{eq:32}
  S_1(\gamma) &= \frac{1}{2} + \oo{\gamma^2} \notag\\
  S_2(\gamma) &= \frac{3}{4} + \oo{\gamma^2}.
\end{align}
And finally,
\begin{align}
  \label{eq:33}
  v_E^2 &= \tr{H^2\rho} - {( \tr{H \rho} )}^2 \notag\\
        &= \frac{4 E_1^2}{\phi^4} \left[ \frac{3}{4} - \frac{1}{2^2} \right] + \oo{\phi^2}\notag\\
        &= 4 {\left( \frac{\pi^2 \hbar^2}{2m L^2} \right)}^2
          {\left( \frac{L}{2\pi \sigma} \right)}^4 \frac{1}{2} + \oo{\phi^2}\notag\\
        &= \frac{1}{2} \frac{\hbar^4}{m^2}\frac{1}{16\sigma^4} + \oo{\phi^2}\notag\\
  v_E &= \frac{\hbar^2}{4 \sqrt{2}m \sigma^2} + \oo{\frac{\sigma}{L} }^2.
\end{align}
Dividing $\hbar/v_E$ by $T_g = \frac{4mL^2}{\hbar \pi}$ then yields \cref{eq:41}.

\section{Uniform Distribution}
\label{sec-1-1}

\begin{figure}
    \centering
    \includegraphics[width=\linewidth]{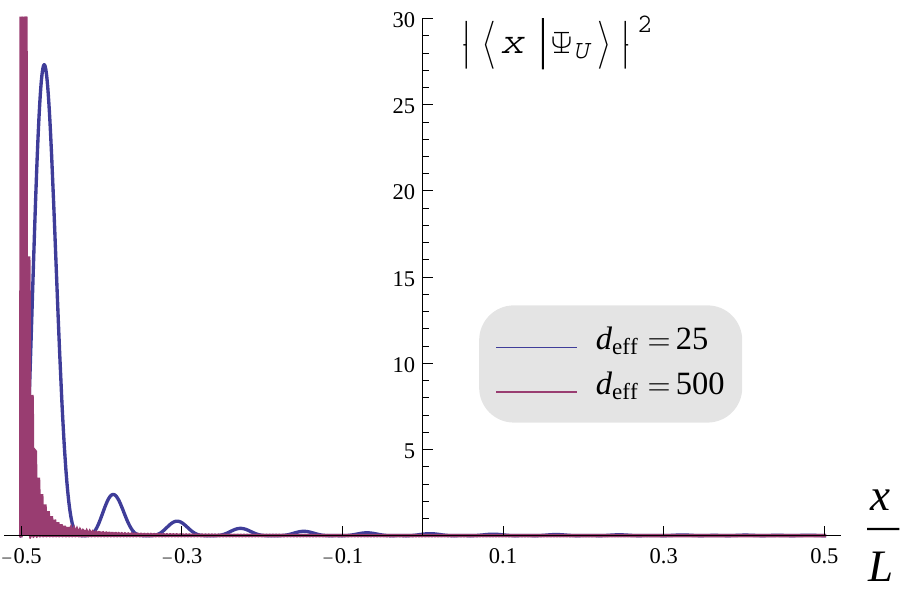}
    \caption{The wave function as defined in (\ref{eq:17}), for different values of $\de$.}
    \label{fig:uni-i-c}
    \centering
    \includegraphics[width=\linewidth]{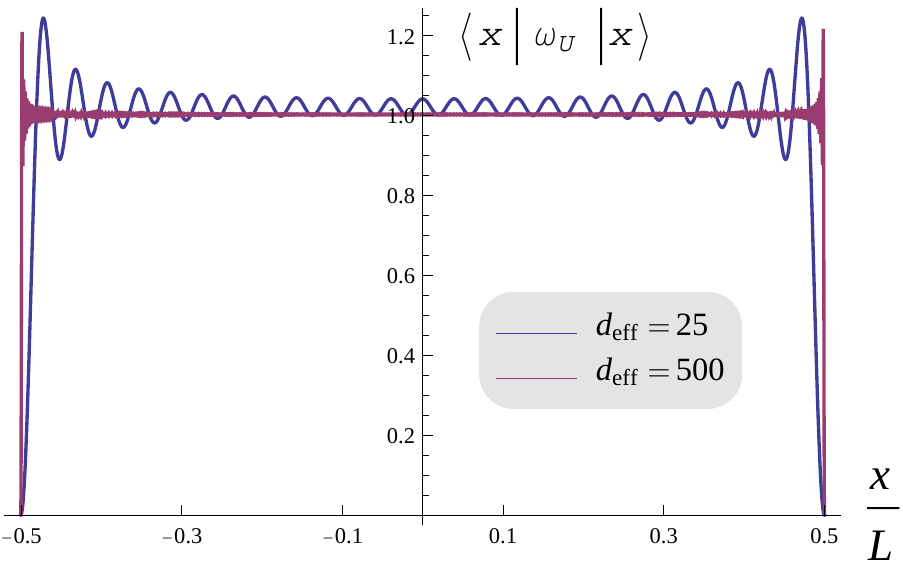}
    \caption{The spatial distribution of the equilibrium state $\omu$ for the initial state in \cref{eq:17}, for different values of $\de$.}
    \label{fig:uni-eq}
\end{figure}
\begin{figure}
    \centering
    \includegraphics[width=\linewidth]{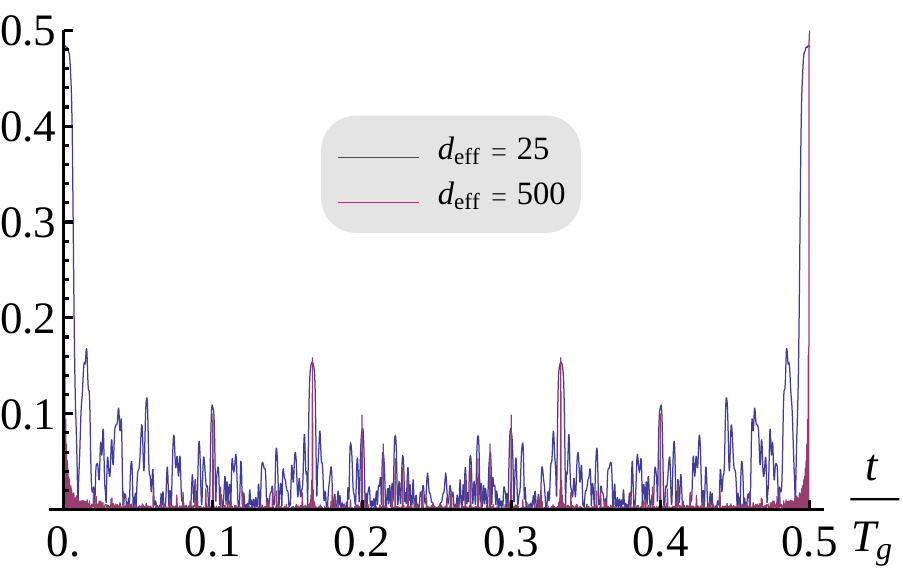}
    \caption{$D_M(\omu,\rou(t))$ with the uniform initial condition for two different $\de$.}
    \label{fig:2}
\end{figure}
\begin{figure}
    \centering
    \includegraphics[width=\linewidth]{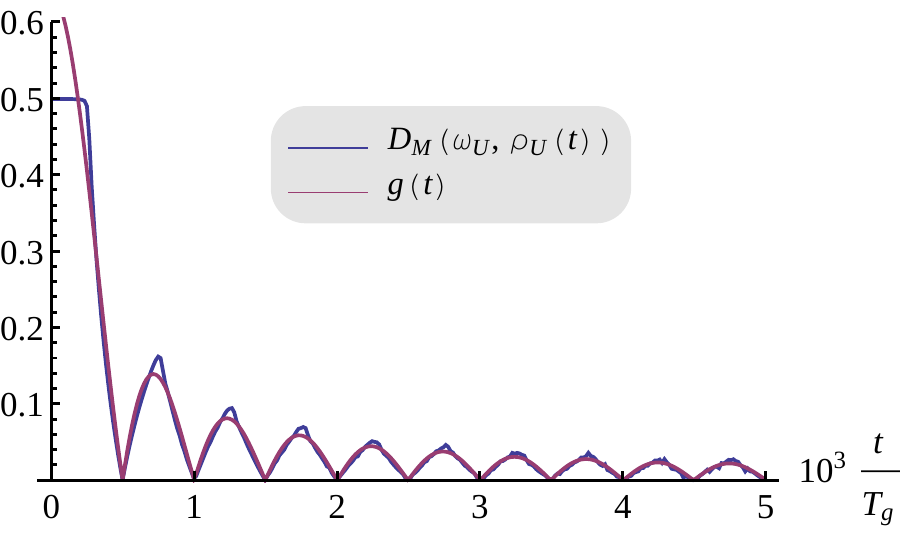}
    \caption{Short time comparison between the numerical computation of $D_M(\omu,\rou(t))$ (for $\de=500$) and the $g(t)$ given in \cref{eq:18}.}
    \label{fig:1}
\end{figure}

The second initial condition considered was the constant distribution over $N$ energy eigenstates, with
\begin{equation}
    \label{eq:17}
    \ket{\psu} = \frac{1}{\sqrt{N}}\sum_{i = 1}^{N} \ket{i},
    \quad \rou = \ket{\psu}\!\!\bra{\psu},
\end{equation}
for some $N \in \Nat$.
This initial condition was chosen because it allows for the exact evaluation of the most relevant part of the sum which defines the distinguishability.
It is trivial to verify that $\de = N$. \Cref{fig:uni-i-c} shows the probability distribution of this state in position space, while \cref{fig:uni-eq} displays the same for its corresponding equilibrium state $\omu = \frac{1}{N} \sum_{n = 1}^{N} \ketbra{n}$.

Since the initial state is clearly concentrated on the left side of the box, we use a slightly different measurement $M = \left\{\Pi_L,\Pi_R \right\}$, which checks on which side of the box the particle is
\begin{equation}
    \label{eq:16}
    \Pi_L = \int_{-\nf[L]{2}}^{0} \ketbra{x} \dif x , \quad \Pi_R = \id - \Pi_L.
\end{equation}
\begin{figure}
    \includegraphics[width=\linewidth]{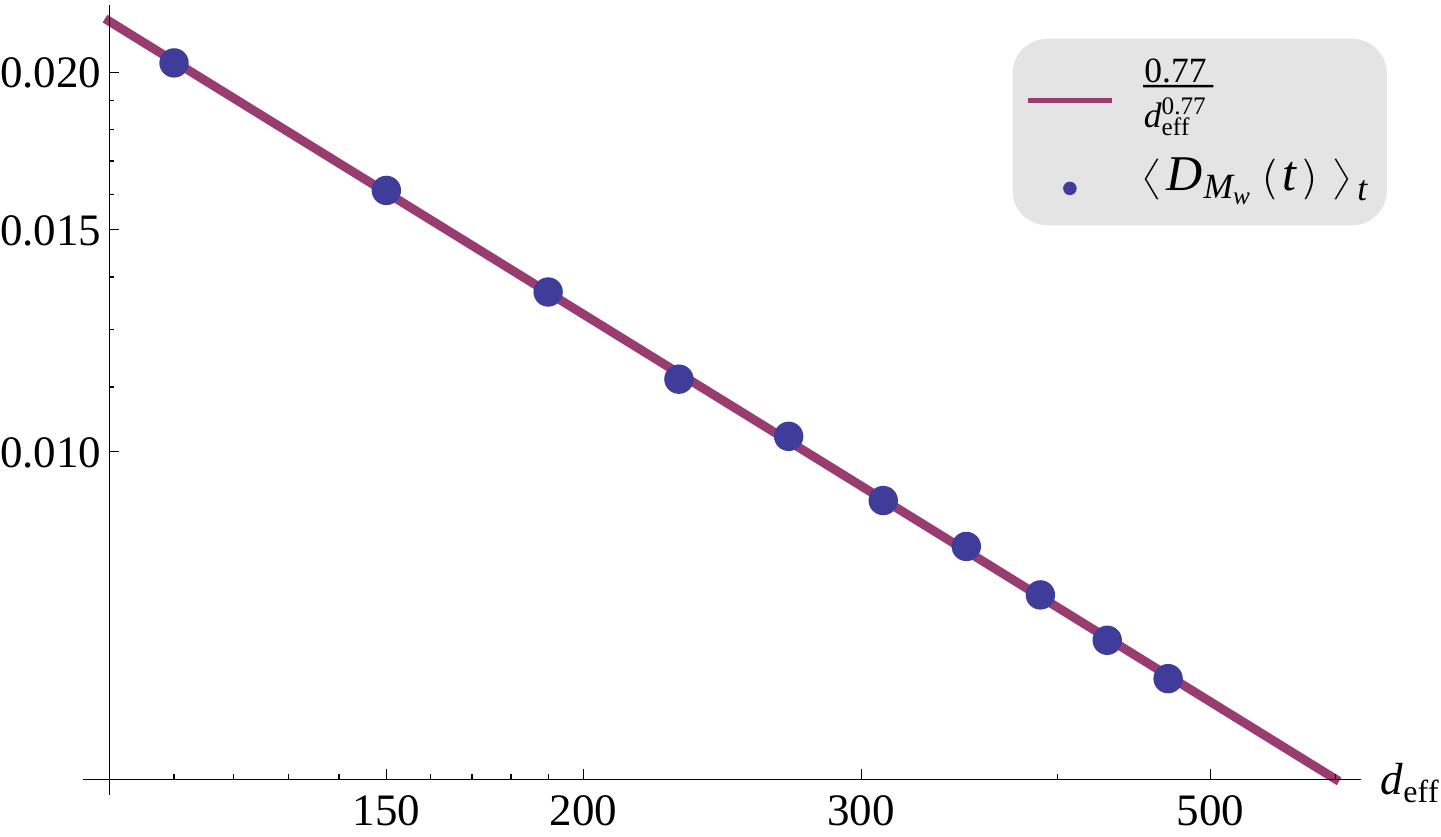}
    \caption{Log-scale plot of the time average of the distinguishability for the uniform initial condition, $D_{M}(\rou(t),\omu)$, as per \cref{eq:8,eq:17,eq:15} for several values of $\de$.}
    \label{fig:7}
\end{figure}

Again, we wish to study the time evolution of the distinguishability according to this measurement, it is now given by the sum
\begin{equation}
    \label{eq:15}
    D_M(\omu,\rou) = \left| \sum_{n\neq j = 1}^{N} %
      \frac{\expo{i (n^2-j^2)\nu t}}{N}
      \braket{n|\Pi_L|j}\right|,
\end{equation}
with $\nu = \frac{\pi^2 \hbar }{2mL^2}$.
This function is displayed on both curves of \cref{fig:2}, and on the blue curve of \cref{fig:1}.
The equilibration behaviour is, again, clearly visible in the high-dimensional limit, and we see (from \cref{fig:1}) that the first zero of the $D_M$ function is a good approximation for the equilibration time scale.

Following similar calculations as those from \cref{sec:equil-time-scale} one finds that, for short times $t \ll T_g$, \cref{eq:15} can be well approximated by the function
\begin{equation}
    \label{eq:18}
    g(t) = \frac{1}{N \pi} \left|\frac{\sin(2N\nu t)}{\sin(\nu t)} - {\cos(\nu t)}\right|.
\end{equation}
\Cref{fig:1} shows how these two functions compare to each other.
To estimate the first zero of $D_M(t)$, we calculate the first zero of $g(t)$, which is exactly at $2\sin(2N\nu t) = \sin(2\nu t)$.
For large $N$, this is approximately
\begin{align}
  \label{eq:23}
  \sin(2N\nu t) &\approx \nu t \notag\\
  t &\approx \frac{2\pi}{4N\nu} = \frac{T_g}{4N},
\end{align}
which yields the time scale
\begin{align}
  \label{eq:24}
  \frac{\tau_{\mbox{\tiny U}}}{T_g} \approx \frac{1}{4 N} = \frac{1}{4 \de},
\end{align}
in good accordance with the time scale results for the Gaussian.

For completeness, \cref{fig:7} shows that the infinite-time average of the distinguishability in this scenario scales with $\de^{-0.77}$, which is very similar to the Gaussian's $\de^{-0.79}$ although the multiplicative constant is larger.
\end{document}